\documentclass[submission,copyright,creativecommons]{eptcs}
\providecommand{\event}{AFL 2023} % Name of the event you are submitting to

\usepackage{iftex}
\usepackage{booktabs}
\usepackage[table]{xcolor}

\usepackage{thmtools}

\usepackage{multirow}
\usepackage{rotating}

\ifpdf
  \usepackage{underscore}         % Only needed if you use pdflatex.
  \usepackage[T1]{fontenc}        % Recommended with pdflatex
\else
  \usepackage{breakurl}           % Not needed if you use pdflatex only.
\fi

\usepackage{complexity}
\usepackage{amsmath}
\usepackage{amsthm}
\usepackage{amssymb}

\usepackage{thmtools}

\usepackage{multirow}
\usepackage{tikz-network}
\usepackage{tikz}
\usetikzlibrary{arrows,automata,positioning,calc,shapes.callouts,shapes.arrows,backgrounds,decorations.pathreplacing,calligraphy}

\usepackage[framemethod=tikz]{mdframed}

\usepackage{ifpdf}
\ifpdf
\fi

\newcommand{\resp}{respectively}
\newcommand{\dfa}{DFA}

\newcommand{\N}{\ensuremath{\mathbb{N}}}

\makeatletter
\newcommand\footnoteref[1]{\protected@xdef\@thefnmark{\ref{#1}}\@footnotemark}
\makeatother

\newtheorem{clm}{Claim}

\newcommand{\mpc}{\texttt{mpc}}
\newcommand{\mpl}{\texttt{mpl}}
\newcommand{\mps}{\texttt{mps}}

\newcommand{\dsc}{\texttt{sc}}

\newtheorem{thm}{Theorem}

\newtheorem{lem}[thm]{Lemma}

\title{On Minimal Pumping Constants for Regular Languages}

\author{Markus Holzer \qquad\qquad Christian Rauch
\institute{Institut f\"ur Informatik, Universit\"at Giessen,
  Arndstr.~2, 35392 Giessen, Germany}
\email{\quad holzer@informatik.uni-giessen.de\quad\qquad
  christian.rauch@informatik.uni-giessen.de} }

\newcommand{\titlerunning}{On Minimal Pumping Constants for Regular Languages}
\newcommand{\authorrunning}{M. Holzer \& C. Rauch}

\hypersetup{
  bookmarksnumbered,
  pdftitle    = {\titlerunning},
  pdfauthor   = {\authorrunning},
  pdfsubject  = {EPTCS},               % Consider adding a more appropriate subject or description
}

\begin{document}

\maketitle

\begin{abstract}
  The study of the operational complexity of minimal pumping constants
  started in [\textsc{J.~Dassow} and \textsc{I.~Jecker}. Operational
  complexity and pumping lemmas. \textit{Acta Inform.}, 59:337--355,
  2022], where an almost complete picture of the operational
  complexity of minimal pumping constants for two different variants
  of pumping lemmata from the literature was given. We continue this
  research by considering a pumping lemma for regular languages that
  allows pumping of sub-words at any position of the considered word,
  if the sub-word is long enough [\textsc{S.~J.~Savitch}. Abstract
  Machines and Grammars. 1982].  First we improve on the simultaneous
  regulation of minimal pumping constants induced by different pumping
  lemmata including Savitch's pumping lemma. In this way we are able
  to simultaneously regulate four different minimal pumping
  constants. This is a novel result in the field of descriptional
  complexity. Moreover, for Savitch's pumping lemma we are able to
  completely classify the range of the minimal pumping constant for
  the operations Kleene star, reversal, complement, prefix- and
  suffix-closure, union, set-subtraction, concatenation, intersection,
  and symmetric difference. In this way, we also solve some of the
  open problems from the paper that initiated the study of the
  operational complexity of minimal pumping constants mentioned above.
%  \keywords{finite automata \and pumping \and regular languages \and finite-state devices \and descriptional complexity.}
\end{abstract}

\section{Introduction}

Pumping lemmata are fundamental to the study of formal languages. An
annotated bibliography on variants of pumping lemmata for regular and
context-free languages is given in~\cite{Ni82}. One variant of the
pumping lemma states that for any regular language~$L$, there exists a
constant~$p$ (depending on~$L$) such that any word~$w$ in the language
of length at least~$p$ can be split into three parts $w=xyz$,
where~$y$ is non-empty, and $xy^tz$ is also in the language, for every
$t\geq 0$---see Lemma~\ref{lem:pumping}. By the contrapositive one can
prove that certain languages are not regular. Since the aforementioned
pumping lemma is only a necessary condition, it may happen that such a
proof fails for a particular language such as, e.g.,
$\{\,a^mb^nc^n\mid\mbox{$m\geq 1$ and
  $n\geq 0$}\,\}\cup\{\,b^mc^n\mid\mbox{$m,n\geq 0$}\,\}$.  The
application of pumping lemmata is not limited to prove non-regularity.
For instance, they also imply an algorithm that decides whether a
regular language is finite or not.  A regular language~$L$ is
\emph{infinite} if and only if there is a word of length at least~$p$,
where~$p$ is the aforementioned constant of the pumping
lemma.\footnote{%
  For the other pumping lemma constants~$p$ considered in this paper,
  the statement on infiniteness can be strengthened to: a regular
  language~$L$ is \emph{infinite} if and only if there is a word of
  length~$\ell$ with $p< \ell\leq 2p$. This also holds true if~$p$
  refers to the deterministic state complexity of a language.} Here a
small~$p$ is beneficial.  Thus, for instance, the question arises on
how to determine a small or smallest value for the constant~$p$ such
that the pumping lemma is still satisfied.

For a regular language~$L$ the value of~$p$ in the above-mentioned
pumping lemma can always be chosen to be the number of states of a
finite automaton, regardless whether it is deterministic or
nondeterministic, accepting~$L$.  Consider the unary
language~$a^na^*$, where all values~$p$ with $0\leq p\le n$ do
\emph{not} satisfy the property of the pumping lemma, but $p=n+1$
does. A closer look on some example languages reveals that sometimes a
much smaller value suffices. For instance, consider the language
$$L=a^* + a^*bb^* + a^*bb^*aa^* + a^*bb^*aa^*bb^*,$$ 
which is accepted by a (minimal) deterministic finite automaton with
five states, the sink state included, but already for $p=1$ the
statement of the pumping lemma is satisfied. It is easy to see that
regardless whether the considered word starts with~$a$ or~$b$, this
letter can be readily pumped. Thus, the minimal pumping constant
satisfying the statement of pumping lemma for the language~$L$ is~$1$,
because the case $p=0$ is equivalent to $L=\emptyset$. This leads to
the notation of a minimal pumping constant for a language~$L$ w.r.t.\
a particular pumping lemma, which is the smallest number~$p$ such that
the pumping lemma under consideration for the language~$L$ is
satisfied.

Recently minimal pumping lemmata constants were investigated from a
descriptional complexity perspective in~\cite{DaJe22}.  Besides basic
facts on these constants for two specific pumping
lemmata~\cite{Br84,HoUl79,Ko97,RaSc59} their relation to each other
and their behaviour under regularity preserving operations was studied
in detail.  In fact, it was proven that for three
natural numbers~$p_1$,\ $p_2$, and~$p_3$ with
$1\leq p_1\leq p_2\leq p_3$, there is a regular language~$L$ over a
growing size alphabet such that $\mpc(L) = p_1$,\ $\mpl(L) = p_2$, and
$\dsc(L) = p_3$, where $\mpc$ ($\mpl$, \resp) refers to the minimal
pumping constant induced by the pumping lemma from~\cite{Ko97}
(from~\cite{Br84,HoUl79,RaSc59}, \resp) and $\dsc$ is the abbreviation
of the deterministic state complexity. This simultaneous regulation of
three measures is novel in descriptional complexity theory. For the
exact statements of the pumping lemmata mentioned above we refer to
Lemma~\ref{lem:pumping} and its following paragraph.  The
\emph{operational complexity of pumping} or \emph{pumping lemmata} for
an $n$-ary regularity preserving operation~$\circ$ undertaken
in~\cite{DaJe22} is in line with other studies on the operational
complexity of other measures for regular languages such as the state
complexity or the accepting state complexity to mention a few. The
operational complexity of pumping is the study of the set
$g_{\circ}(k_1 ,k_2 ,\ldots,k_n )$ of all numbers~$k$ such that there
are regular languages $L_1, L_2,\ldots ,L_n$ with minimal pumping
complexity $k_1,k_2,\ldots,k_n$, \resp, and the language
$L_1\circ L_2\circ\cdots\circ L_n$ has minimal pumping complexity~$k$.
In~\cite{DaJe22} a complete picture for the operational complexity
w.r.t.\ the pumping lemma from~\cite{Ko97} (measure~$\mpc$) for the
operations Kleene closure, complement, reversal, prefix and
suffix-closure, circular shift, union, intersection, set-subtraction,
symmetric difference, and concatenation was given---see
Table~\ref{tab:results} on page~\pageref{tab:results}.  However, for
the pumping lemma from~\cite{Br84,HoUl79,RaSc59} (measure~$\mpl$) some
results from~\cite{DaJe22} are only partial (set-subtraction and
symmetric difference) and others even remained open (circular shift
and intersection); for comparison see the table mentioned above. The
behaviour of these measures differ with respect to finiteness/infinity
of ranges, due to the fact that for the pumping lemma
from~\cite{Br84,HoUl79,RaSc59} the pumping has to be done within a
prefix of bounded length.

This is the starting point of our investigation.  As a first step we
improve on the above mentioned result on the simultaneous regulation
of minimal pumping constants showing that already a binary language
suffices. If we additionally also consider a fourth measure ($\mps$)
induced by the pumping lemma from~\cite{Sa82}, we obtain a similar
result for a quinary language. Thus, we are able to regulate four
descriptional complexity measures simultaneously on a single regular
language. Savitch's pumping lemma allows pumping of sub-words at any
position of the considered word, if the sub-word is long enough---see
Lemma~\ref{lem:pumping-xy-length}.  Moreover, the outcome of our study
on the operational complexity of pumping presents a comprehensive view
for the previously mentioned operations. In passing, we can also solve
all the partial and open problems from~\cite{DaJe22}, completing the
overall picture for the three pumping lemmata in question---see the
gray shaded entries in Table~\ref{tab:results} on
page~\pageref{tab:results}. This provides a full understanding of the
operational complexity of these pumping lemmata. it is worth
mentioning that the obtained result are very specific to the
considered pumping lemmata---compare with~\cite{GHR23,HoRa23a} where
descriptional and computational complexity aspects of Jaffe's pumping
lemma~\cite{Ja78} are considered. For instance, the simultaneous
regulation of pumping constants involving those satisfying Jaffe's
pumping lemma seems to be much more complicated, since only the
deterministic state complexity can serve as an upper bound, while the
nondeterministic state complexity becomes incomparable. Due to space
constraints almost all proofs are omitted; they can be
found in the full version of this paper.

\section{Preliminaries}

We recall some definitions on finite automata as contained
in~\cite{Ha78}. Let~$\Sigma$ be an alphabet. Then, as usual~$\Sigma^*$
refers to the set of all words over the alphabet~$\Sigma$, including
the empty word~$\lambda$, and $\Sigma^{\leq k}$ denotes the set of all
words of length at most~$k$. For a word $w=a_1a_2\ldots a_n\in \Sigma^*$ and a
natural number $k\geq 1$ we refer to the word $a_1a_2\ldots a_k$, if
$k\leq n$, and $a_1a_2\ldots a_n$, otherwise, as the
\emph{$k$-prefix of~$w$}. If $k=0$, then~$\lambda$ is the unique
$0$-prefix of any word. Analogously one can define the
\emph{$k$-suffix of a word~$w$}.

A \emph{deterministic finite automaton} (\dfa) is a quintuple
$A=(Q,\Sigma,{}\cdot{},q_0,F)$, where~$Q$ is the finite set of
\textit{states}, $\Sigma$ is the finite set of \textit{input symbols},
$q_0\in Q$ is the \textit{initial state}, $F\subseteq Q$ is the set of
\textit{accepting states}, and the \textit{transition
  function}~$\cdot$ maps $Q\times\Sigma$ to~$Q$. The \emph{language
  accepted} by the \dfa~$A$ is defined as
$L(A) =\{\,w\in \Sigma^*\mid\mbox{$q_0\cdot w\in F$}\,\}$, where the
transition function is recursively extended to a mapping
$Q\times\Sigma^*\rightarrow Q$ in the usual way.  Finally, a finite
automaton is \emph{unary} if the input alphabet~$\Sigma$ is a
singleton set, that is, $\Sigma=\{a\}$, for some input symbol~$a$.
The \emph{deterministic state complexity of a finite automaton}~$A$
with state set~$Q$ is referred to as $\dsc(A):=|Q|$ and the
\emph{deterministic state complexity of a regular language}~$L$ is
defined as
$$\dsc(L)=\min\{\,\dsc(A)\mid\mbox{$A$ is a \dfa\ accepting~$L$, i.e., $L=L(A)$}\,\}.$$

A finite automaton is \emph{minimal} if its number of states is
minimal with respect to the accepted language. It is well known that
each minimal \dfa\ is isomorphic to the \dfa\ induced by the
Myhill-Nerode equivalence relation.  The \emph{Myhill-Nerode}
equivalence relation~$\sim_L$ for a
language~$L\subseteq\Sigma^*$ is defined as follows: for
$u,v\in \Sigma^*$ let $u\sim_Lv$ if and only if $uw\in L\iff vw\in L$,
for all~$w\in \Sigma^*$. The equivalence class of~$u$ is referred to
as~$[u]_L$ or simply~$[u]$ if the language is clear from the context
and it is the set of all words that are equivalent to~$u$ w.r.t.\ the
relation~$\sim_L$, i.e., $[u]_L=\{\,v\mid u\sim_L v\,\}$.

Regular languages satisfy a variety of different pumping lemmata---for
a comprehensive list of pumping or iteration lemmata we refer
to~\cite{Ni82}. A well known pumping lemma variant can be found
in~\cite[page~70, Theorem~11.1]{Ko97}:

\begin{lem}\label{lem:pumping}
  Let~$L$ be a regular language over~$\Sigma$. Then, there is a
  constant~$p$ (depending on~$L$) such that the following holds: If
  $w\in L$ and $|w|\geq p$, then there are words $x\in \Sigma^*$,\
  $y\in \Sigma^+$, and $z\in \Sigma^*$ such that $w = xyz$ and
  $xy^tz \in L$ for~$t\geq 0$---it is then said that~$y$ can be
  \emph{pumped} in~$w$. Let~$\mpc(L)$ denote the smallest number~$p$
  satisfying the aforementioned statement.
\end{lem}

The above lemma can be slightly modified with the condition
$|xy|\leq p$, which can be found in~\cite[page~119, Lemma~8]{RaSc59},
\cite[page~252, Folgerung~5.4.10]{Br84}, and~\cite[page~56,
Lemma~3.1]{HoUl79}. Analogously, to $\mpc$ one defines~$\mpl(L)$, as
the smallest number~$p$ satisfying the statement of the modified
pumping lemma.

Recently, pumping lemmata were considered in~\cite{DaJe22}, where 
besides some simple facts
such as
\begin{enumerate}
\item $\mpc(L)=0$ if and only if $\mpl(L)=0$ if and only if
  $L=\emptyset$,
\item for every non-empty finite language~$L$ we have
  $\mpc(L)=\mpl(L)= 1+\max\{\,|w|\mid w\in L\,\}$,
\item $\mpc(L)=1$ implies $\lambda\in L$, and
\item if $\mpl(L)=1$, then~$L$ is suffix closed,\footnote{A
    language~$L\subseteq\Sigma^*$ is \emph{suffix closed} if
    $L=\{\,x\mid\mbox{$yx\in L$, for some $y\in\Sigma^*$}\,\}$, i.e.,
    the word~$x$ is a member of~$L$ whenever $yx$ is in~$L$, for
    some~$y\in\Sigma^*$.}
\end{enumerate}
also the inequalities
$$\mpc(L)\leq\mpl(L)\leq\dsc(L)$$
and results on the operational complexity w.r.t.\ these minimal
pumping constants were shown. The upper bound on the
minimal pumping constants by the deterministic state complexity is
obvious. Moreover, in~\cite{DaJe22} it was also proven that for three
natural numbers~$p_1$,\ $p_2$, and~$p_3$ with
$1\leq p_1\leq p_2\leq p_3$, there is a regular language~$L$ such that
$\mpc(L) = p_1$,\ $\mpl(L) = p_2$, and $\dsc(L) = p_3$.
The witness language to prove this result is in almost all
cases, except for $p_2=p_3$,
$$L= b^{p_1-1}(a^{p_2-p_1+1})^* + c_1^* + c_2^* + \cdots + c_{p_3-p_2-1}^*,$$
while for the remaining case a unary language is given. Hence,~$L$ is
a language over an alphabet of \emph{growing size}. We improve on this
result, showing that already a \emph{binary} language can be
used. Moreover, we also fix a simple flaw\footnote{%
  For $1\leq p_1\leq p_2\leq p_3$ let
$$L= b^{p_1-1}(a^{p_2-p_1+1})^* + c_1^* + c_2^* + \cdots c_{p_3-p_2-1}^*,$$
over the alphabet $\{a,b\}\cup\{\,c_i\mid 1\leq i\leq
p_3-p_2-1\}$. For $p_1=p_2=1$ consider the above given language. In
case $p_3=2$ we get the language $L=a^*$ over the
alphabet~$\{a,b\}$, which requires a minimal \dfa\ with~$2$ states and
in case $p_3\geq 3$ we have $L=a^*+c_1^*+c_2^*+\cdots + c_{p_3-2}^*$
over the alphabet $\{a,b\}\cup\{\,c_i\mid 1\leq i\leq p_3-2\}$. Note
that $p_3-2\geq 1$ since $p_3\geq 3$ and therefore the latter set in
the union of the alphabet letters is non-empty. Thus, the minimal
\dfa\ accepting the language~$L$ has~$p_3+1$ states, which are
responsible for the Myhill-Nerode equivalence classes
$[\lambda]=\{\lambda\}$,\ $[a]=a^+$,\ $[c_1]=c_1^+$,\ $[c_2]=c_2^+$,
\ldots, $[c_{p_3-2}]=c_{p_3-2}^+$, and finally the equivalence class
$[b]=\{\,w\mid\mbox{$w\in b^+$ or $w$ contains at least two different
  letters}\,\}$. Observe, that all equivalence classes are accepting,
except the class~$[b]$, which represents the non-accepting sink
state. Hence in case $p_1=p_2=1$ and $p_2<p_3$ the statement on the
number of states of the minimal \dfa\ accepting the language~$L$
presented in~\cite{DaJe22} is off by one state. The claims on the
minimal pumping constants $\mpc$ and~$\mpl$ for~$L$ are correct. Note
that the case $p_3=1$ is shown in~\cite{DaJe22} with the help of a
unary language.}  on the size of the automaton in case $p_1=p_2=1$ and
$p_2<p_3$ in the original proof given in~\cite{DaJe22}.

\begin{thm}\label{thm:mpc-leq-mpl-leq-sc-binary}
  Let $p_1$,\ $p_2$, and~$p_3$ be three natural numbers with
  $1\leq p_1\leq p_2\leq p_3$. Then, there is a regular language~$L$
  over a \emph{binary alphabet} such that $\mpc(L) = p_1$,\
  $\mpl(L) = p_2$, and $\dsc(L) = p_3$.
\end{thm}

\begin{proof}
  First we define some useful languages. For $k \geq 1$ let 
  $$B^{(+)}_k =
  \begin{cases}
  	b^+(a^*b^*)^{(k-1)/2}, & \mbox{if $k$ is odd,}\\
  	b^+(a^*b^*)^{(k-2)/2}a^*, & \mbox{if $k$ is even,}
  \end{cases}$$
  and
  $$B^{(*)}_k =
  \begin{cases}
  	b^*(a^*b^*)^{(k-1)/2}, & \mbox{if $k$ is odd,}\\
  	b^*(a^*b^*)^{(k-2)/2}a^*, & \mbox{if $k$ is even,}
  \end{cases}$$
  be languages over the alphabet~$\Sigma=\{a,b\}$.
  Observe that in all cases there are $k-1$ alternations between the blocks. Thus, e.g., $B^{(*)}_3=b^*a^*b^*$ and $B^{(+)}_4=b^+a^*b^*a^*$. In case $k=0$ the languages~$B^{(+)}_k$ and~$B^{(*)}_k$ are set to~$\emptyset$.
  Observe that~$B_k^{(+)}+\lambda$ is not equal to~$B_k^{(*)}$.

  Now we are ready for the proof. We distinguish whether~$p_2=1$ (this
  implies that $p_1=p_2=1$) or~$p_2=p_3$ (which implies
  $p_1\leq p_2=p_3$) or $p_2\notin\{1,p_3\}$.
  \begin{enumerate}
  \item Case $p_1=p_2=1$.  For~$p_3=1,2$ we simply use the \dfa s
    accepting the languages~$\Sigma^*$, $a^*$, \resp,
    for~$\Sigma=\{a,b\}$ being the input alphabet of those automata.
    For~$p_3\ge 3$ we observe that the languages~$B^{(*)}_{p_3-1}$
    fulfill~$\mpc(B^{(*)}_{p_3-1})=\mpl(B^{(*)}_{p_3-1})=p_1=p_2=1$
    since each accepted word can be pumped by its first
    letter. Additionally those languages are accepted by the \dfa~$A$
    shown in Figure~\ref{fig:m1-m2-1}---the non-accepting sink state
    is not shown.
    \begin{figure}[!htb]
      \begin{center}
        \begin{tikzpicture}[->,>=stealth,shorten >=1pt,auto,node distance=0.9cm,
	semithick,scale=0.7, every node/.style={scale=0.7},initial text=]
	\tikzstyle{every state}=[fill={rgb:black,1;white,10},ellipse]
	
	\node[initial, state,accepting] (p0) [thick] {$q_0$};						
	\node[state,accepting] (p1) [right = 1.4cm of p0] [thick]{$q_1$};
	\node[state,accepting] (p2) [right = 1.4cm of p1] [thick]{$q_{2}$};
	\node[state,accepting] (p3) [right = 1.4cm of p2] [thick]{$q_3$};
	\node[state,accepting] (p4) [right = 2.4cm of p3] [thick]{$q_{p_3-2}$};

	\path 
	
	(p0) edge[] node[midway,above]{$a$} (p1)
	(p1) edge[] node[midway,above]{$b$} (p2)
	(p2) edge[] node[midway,above]{$a$} (p3)
	(p3) edge[dashed] node[midway,above]{} (p4)
	(p4) edge[loop above] node[midway,above]{$a$} (p4)
	
	(p0) edge[loop above] node[midway,above]{$b$} (p0)
	(p1) edge[loop above] node[midway,above]{$a$} (p1)
	(p2) edge[loop above] node[midway,above]{$b$} (p2)
	(p3) edge[loop above] node[midway,above]{$a$} (p3)

	;
\end{tikzpicture}	
      \end{center}
    \caption{The automaton~$A$ for~$p_1=p_2=1$ and~$p_3-1$ even, where the non-accepting sink state~$q_{p_3-1}$ and all transitions to it are not shown. Recall, that the letter on the transition to~$q_{p_3-2}$ depend on the parity of~$p_3-2$.}
      \label{fig:m1-m2-1}
    \end{figure}
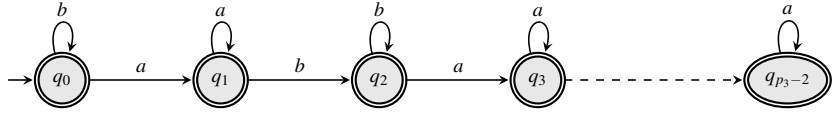
    It is not hard to see that for each state of~$A$ there is a unique
    shortest word that maps the state onto the non-accepting state.
    Therefore we have that~$A$ is minimal
    and~$\dsc(B^{(*)}_{p_3-1})=\dsc(A)=p_3$.
		
  \item Case~$p_1\leq p_2=p_3$. In this case we define the
    unary \dfa\
    $$A=(\{q_0,q_1,\ldots,q_{p_3-1}\},\{a\},{}\cdot_A{},q_0,\{q_{p_1-1}\}),$$ with~$q_i\cdot_A a = q_{i+1\bmod p_3}$, for $0\leq i\leq p_3-1$.
    By inspecting Figure~\ref{fig:thm-mpc-x-mpl-eq-sc} which shows~$A$
    it is not hard to see
    that~$L(A)=\{\,a^{p_2\cdot i+p_1-1}\mid i\ge 0\,\}$ and that~$A$
    is already minimal; thus $\dsc(L(A))=p_3$.
    \begin{figure}[!htb]
      \begin{center}
        \begin{tikzpicture}[->,>=stealth,shorten >=1pt,auto,node distance=0.9cm,
	semithick,scale=0.7, every node/.style={scale=0.7},initial text=]
	\tikzstyle{every state}=[fill={rgb:black,1;white,10},ellipse]
	
	\node[initial, state] (p0) [thick] {$q_0$};						
	\node[state] (p1) [right = 1.4cm of p0] [thick]{$q_1$};
	\node[state,accepting] (p2) [right = 1.4cm of p1] [thick]{$q_{p_1-1}$};
	\node[state] (p3) [right = 1.4cm of p2] [thick]{$q_{p_1}$};
	\node[state] (p4) [right = 2.4cm of p3] [thick]{$q_{p_3-1}$};

	\path 
	
	(p0) edge[] node[midway,above]{$a$} (p1)
	(p1) edge[dashed] node[midway,above]{$a$} (p2)
	(p2) edge[] node[midway,above]{$a$} (p3)
	(p3) edge[dashed] node[midway,above]{$a$} (p4)
	(p4) edge[bend right=15] node[midway,above]{$a$} (p0)
	;
\end{tikzpicture}	
      \end{center}
      \caption{The unary automaton~$A$ for $p_1<p_2=p_3$.}
      \label{fig:thm-mpc-x-mpl-eq-sc}
    \end{figure}
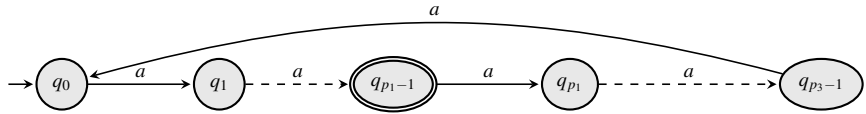
    So every word in the language~$L(A)$ that has length greater or
    equal~$p_1$ contains the sub-word~$a^{p_2}$ which implies that it
    is pumpable. On the other hand the word~$a^{p_1-1}$ cannot be
    pumped since it is the shortest accepting word; hence it cannot be
    shortened by pumping.  Therefore~$\mpc(L(A))=p_1$
    and~$\mpl(L(A))=p_2$.
	
	\item Case~$p_2\notin\{1,p_3\}$. 
	We define the language
	$$L= b^{p_1-1}(a^{p_2-p_1+1})^* (B^{(+)}_{p_3-p_2-1} + \lambda).$$
	This language is accepted by the \dfa\ shown in Figure~\ref{fig:thm-mpc-l-mpl-l-sc}; again the non-accepting sink state is not shown.
	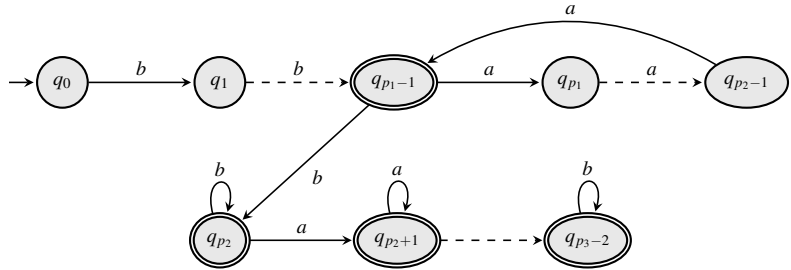
\begin{figure}[!htb]
		\begin{center}
			\begin{tikzpicture}[->,>=stealth,shorten >=1pt,auto,node distance=0.9cm,
	semithick,scale=0.7, every node/.style={scale=0.7},initial text=]
	\tikzstyle{every state}=[fill={rgb:black,1;white,10},ellipse]
	
	\node[initial, state] (p0) [thick] {$q_0$};						
	\node[state] (p1) [right = 1.4cm of p0] [thick]{$q_1$};
	\node[state,accepting] (p2) [right = 1.4cm of p1] [thick]{$q_{p_1-1}$};
	\node[state] (p3) [right = 1.4cm of p2] [thick]{$q_{p_1}$};
	\node[state] (p4) [right = 1.4cm of p3] [thick]{$q_{p_2-1}$};

	\node[state,accepting] (p5) [below = 1.4cm of p1] [thick]{$q_{p_2}$};
	\node[state,accepting] (p6) [right = 1.4cm of p5] [thick]{$q_{p_2+1}$};
	\node[state,accepting] (p7) [right = 1.4cm of p6] [thick]{$q_{p_3-2}$};

	\path 
	
	(p0) edge[] node[midway,above]{$b$} (p1)
	(p1) edge[dashed] node[midway,above]{$b$} (p2)
	(p2) edge[] node[midway,above]{$a$} (p3)
	(p3) edge[dashed] node[midway,above]{$a$} (p4)
	(p4) edge[bend right] node[midway,above]{$a$} (p2)
	
	(p2) edge[] node[midway]{$b$} (p5)
	(p5) edge[] node[midway,above]{$a$} (p6)
	(p6) edge[dashed] node[midway,above]{} (p7)
	
	(p5) edge[loop above] node[midway,above]{$b$} (p5)
	(p6) edge[loop above] node[midway,above]{$a$} (p6)
	(p7) edge[loop above] node[midway,above]{$b$} (p7)

	;
\end{tikzpicture}	
		\end{center}
        \caption{The automaton~$A$ for the language~$L$ in case $p_3-p_2-1$ is odd, where the non-accepting sink state~$q_{p_3-1}$ and all transitions to it are not shown. In case $p_3-p_2-1$ is even the lower sub-chain of states looks similar by alternatively reading~$a$'a and~$b$'s, has appropriate self-loops on the states, and end with the letter~$a$.}
		\label{fig:thm-mpc-l-mpl-l-sc}
	\end{figure}
	Observe that each state~$q_i$, for~$i\in\{0,1,\dots,p_2-1\}\setminus\{p_1-1\}$, is only mapped by one letter onto a state that is unequal to the sink state while this is \emph{not} true for each state~$q_i$, for~$i\in\{p_2,p_2+1,\dots,p_3-3,p_1-1\}$.
	Then one can easily prove that this \dfa\ is minimal. Thus, the automaton~$A$ has~$p_3$ states.
	Further we observe that the word~$b^{p_1-1}$ is in~$L$ but it cannot be pumped since no shorter word is in~$L$. Therefore, $\mpc(L)\ge p_1$. Additionally we observe that $w\in b^{p_1-1}(a^{p_2-p_1+1})^+$ is a word in~$L$ which is only pumpable by~$a^{p_2-p_1+1}$. Since the shortest prefix of~$w$ that ends with~$a^{p_2-p_1+1}$ has length~$p_2$ we obtain that~$\mpl(L)\ge p_2$.
	Clearly we can pump all words in~$b^{p_1-1}(a^{p_2-p_1+1})^+B^{(+)}_{p_3-p_2-1}$ in the same way which implies that none of these words has an impact on~$\mpc(L)$ and~$\mpl(L)$.
	Last we see that all words in~$b^{p_1-1}B^{(+)}_{p_3-p_2-1}$ can be pumped by their first letter or by their $(p_1+1)$th letter, \resp, for~$p_1=p_2$ and~$p_1<p_2$.
	So we obtain that all words in~$L$ which have length at least~$p_1$ can be pumped by a sub-word in their prefix of length at most~$p_2$. Thus, we have~$\mpc(L)=p_1$ and~$\mpl(L)=p_2$.
      \end{enumerate}
      This completes the construction and proves the stated claim for languages over a binary alphabet.
\end{proof}

The previous theorem is best possible w.r.t.\ the alphabet size,
because for unary languages there are infinitely many combinations of
minimal pumping constants like, e.g., $\mpc(L)=\mpl(L)=1$ and
$\dsc(L)\geq 2$, which cannot be achieved by any unary
language~$L$. This is due to the fact that if $\mpl(L)=1$, then the
language~$L$ is suffix-closed, and~$\{a\}^*$ is the only suffix-closed
unary language. It is not hard to prove that
Theorem~\ref{thm:mpc-leq-mpl-leq-sc-binary} is also valid if the
nondeterministic state complexity instead of the deterministic state
complexity is considered.

\section{Results on Sub-Word Pumping}

Let us first introduce a pumping lemma which is a straight forward
generalization of Lemma~\ref{lem:pumping} with the additional
$|xy|\leq p$ condition. The lemma can be found in~\cite[page~49,
Theorem~3.10]{Sa82} and reads as follows--- roughly speaking, this
pumping lemma allows pumping of sub-words, whose length is large
enough, at any position of the considered word; hence we sometimes
speak of \emph{sub-word pumping}.

\begin{lem}\label{lem:pumping-xy-length}
  Let~$L$ be a regular language over~$\Sigma$. Then there is a
  constant~$p$ (depending on~$L$) such that the following holds: If
  $\tilde{w}=uwv\in L$ and $|w|\geq p$, where~$u$ and~$v$ are any (possibly
  empty) words, then there are words $x\in \Sigma^*$,\
  $y\in \Sigma^+$, and $z\in \Sigma^*$ such that $w = xyz$,\
  $|xy|\leq p$, and $uxy^tzv\in L$ for $t\geq 0$.
\end{lem}

Similarly as for the aforementioned pumping lemmata, one can define
the minimal pumping constant~$\mps(L)$, for a regular language, as the
smallest number~$p$ that satisfies the condition of
Lemma~\ref{lem:pumping-xy-length} when considering~$L$. Observe, that
the condition of the lemma requires that any sub-word that is long
enough can be pumped.

\subsection{Comparing $\mps$ to Other Minimal Pumping Constants}

We first prove some basic properties:

\begin{lem}\label{lem:mps-zero-one}
	Let~$L$ be a regular language over~$\Sigma$. Then
	\begin{itemize}
		\item $\mps(L)=0$ if and only if~$L=\emptyset$, and 
		\item $\mps(L)=1$, implies that~$L$ is prefix- and
                  suffix-closed.\footnote{%
                    Moreover, $\mps(L)=1$, also implies that~$L$ is
                    factor-closed.  A regular language~$L$ is
                    \textit{factor-closed} if~$L$ contains all factors
                    of all words~$w\in L$. We
                    call~$w_{i_1}w_{i_2}\ldots w_{i_k}$ a
                    \textit{factor} of the word~$w_1w_2\ldots w_n$
                    if~$1\le i_1<i_2<\cdots<i_k\le n$ are natural
                    numbers.}
	\end{itemize}
\end{lem}

\begin{proof}
	First we observe that there are no words~$x\in \Sigma^*$,\
	$y\in \Sigma^+$, and $z\in \Sigma^*$ such that $|xy|\leq 0$.
	This implies directly that the statement of Lemma~\ref{lem:pumping-xy-length}
	is fulfilled for~$p=0$ and the language~$L$ if and only if~$L=\emptyset$.
	Next we have that~$\mps(L)=1$ implies that for all~$w$ with~$|w|\ge 1$
	and all words~$u,v\in \Sigma^*$ such that~$\tilde{w}=uwv\in L$ there are
	words $x\in \Sigma^*$,\
	$y\in \Sigma^+$, and $z\in \Sigma^*$ such that $w = xyz$,\
	$|xy|\leq 1$, and $uxy^tzv\in L$ for $t\geq 0$.
	In especially this holds for~$w\in \Sigma$ which implies that~$y=w$.
	Since~$uxy^0zv=uxzv\in L$ for all letters~$y=w\in\Sigma$ and all
	(possibly empty) words~$u$ and~$v$ we obtain
	that each word~$\tilde{w}$ of~$L$ can be pumped by each of its letters, i.e,
	by each letter of each prefix and each suffix of~$\tilde{w}$.
	Hence, $L$ is prefix- and suffix-closed.
\end{proof}

Next we want to compare~\mps\ with the other minimal pumping constants
considered in~\cite{DaJe22}. We find the following situation---similarly as
in Theorem~\ref{thm:mpc-leq-mpl-leq-sc-binary} the nondeterministic
state complexity is also an upper bound:

\begin{thm}\label{thm:mps-dsc}
  Let~$L$ be a regular language~$L$ over~$\Sigma$. Then 
  $\mpc(L)\leq \mpl(L)\le \mps(L)\leq \dsc(L)$.
\end{thm}

\begin{proof}
  It suffices to show $\mpl(L)\leq\mps(L)\leq\dsc(L)$. For the
  first inequality observe that if we set~$u=v=\lambda$ in
  Lemma~\ref{lem:pumping-xy-length} we obtain statement of
  Lemma~\ref{lem:pumping} with the additional length condition
  $|xy|\leq p$, which implies that $\mpl(L)\leq\mps(L)$.  Finally, the
  $\dsc(L)$ upper bound is immediate by the proof of the lemma given
  in~\cite[page~49, Theorem~3.10]{Sa82}.
\end{proof}

Now the question arises whether we can come up with a similar result
as stated in Theorem~\ref{thm:mpc-leq-mpl-leq-sc-binary}, but now also
taking the minimal pumping constant w.r.t.\
Lemma~\ref{lem:pumping-xy-length} into account. The following Theorem
will be very useful for this endeavor; a similar statement was shown
in~\cite{HoRa23a} for the minimal pumping constant w.r.t.\ Jaffe's
pumping lemma~\cite{Ja78}, a pumping lemma that is necessary and
sufficient for regular languages.

\begin{thm}\label{thm:loops-only-decrease-mpx}
  Let~$A=(Q,\Sigma,{}\cdot_A{},q_0,F)$ be a minimal \dfa,
  state~$q\in Q$, and letter $a\in\Sigma$. Define the finite automaton
  $B=(Q,\Sigma,{}\cdot_B{},q_0,F)$ with the transition
  function~$\cdot_B$ that is equal to the transition function
  of~$\cdot_A$, except for the state~$q$ and the letter~$a$, where
  $q\cdot_B a=q$. Then, $K(L(B))\leq K(L(A))$
  for~$K\in\{\mpc,\mpl,\mps\}$.
\end{thm}

\begin{proof}%{theorem}{thm:loops-only-decrease-mpx}
  Obviously we have that in each word of the form~$w=xaz$
  with~$q_0\cdot_B x=q$ the~$(|x|+1)$st letter can be pumped, because
  by construction
  $$w=xazv\in L(B)\quad\mbox{if and only if}\quad xa^tzv\in L(B),$$
  for all $t\geq 0$ and each~$v\in\Sigma^*$. On the other hand the
  change of the~$a$-transition of~$q$ does not affect all other words
  not satisfying the above property. On these words the pumping is
  that of the pumping induced by the device~$A$. Thus, we conclude
  that the three mentioned minimal pumping constants for the
  language~$L(B)$ are bounded by the according ones of~$A$.
\end{proof}

	Observe, that the statement of Lemma~\ref{lem:pumping-xy-length} for
	the constant~$n$ can also be understood as follows: for each word~$\tilde{w}$
	in~$L$ and each sub-word~$w$ of~$\tilde{w}$ with length at least~$n$ there is
	a sub-word~$y$ of~$w$ such that~$y$ can be pumped in~$\tilde{w}$.  We
	will use this alternative version of Lemma~\ref{lem:pumping-xy-length}
	in the lemmata to come without further notice.
	
\begin{thm}\label{thm:one-to-rule-em-all}
	Let $p_1$,\ $p_2$,\ $p_3$, and~$p_4$ be four natural numbers with $1\leq p_1\leq p_2\leq p_3\le p_4$. Then, there is a regular language~$L$ over a
	\emph{quinary alphabet} such that
	$\mpc(L) = p_1$,\ $\mpl(L) = p_2$, $\mps(L) = p_3$, and $\dsc(L)=p_4$\ holds.
\end{thm}

\begin{proof}
	By taking an intense look at the constructions shown in the proof of Theorem~\ref{thm:mpc-leq-mpl-leq-sc-binary} we observe that~$\mpl(L)=\mps(L)$ holds
	for all used languages~$L$.
	Therefore we safely assume for the rest of the proof that~$p_2<p_3$.
	On the other hand we distinguish for the proof whether~$p_3\le p_4-1$ or~$p_3=p_4$.
	In the former case we additionally differ between~$p_1=1$ or~$p_1\ge 2$.
	Since the constructions in all cases are adaptions of the case~$p_3\le p_4-1$ and~$p_1\ge 2$ we give all constructions next.
	
	We define the automaton
	$A=(\{q_0,q_1,\dots,q_{p_4-1}\},\{a,b,c,d,e\},{}\cdot{},q_0,\{q_{p_1-1}\}\cup F)$
	with the state set
        $F=\{\,q_i\mid p_3\le i\le p_4-2\,\}$, if~$p_1=1$, and $F= 
		\{\,q_i\mid p_3-1\le i\le p_4-2\,\}$, otherwise.
	%$$F=\begin{cases}
	%	\{\,q_i\mid p_3\le i\le p_4-2\,\}, & \mbox{if~$p_1=1$},\\
	%	\{\,q_i\mid p_3-1\le i\le p_4-2\,\}, & \mbox{otherwise.}
	%\end{cases}
	%$$
	The transition function of~$A$ depends on the
	relation of~$p_3$ and~$p_4$.
	For~$p_3=p_4$ we set
	\begin{align*}
		q_{2i}\cdot a & =q_{2i+1}, & \mbox{for~$0\le i\le (p_3-2)\div 2$,}\\
		q_{2i+1}\cdot c & =q_{2i+2}, & \mbox{for~$0\le i\le (p_3-3)\div 2$,}\\
		q_i\cdot b & =q_{i-1}, & \mbox{for~$1\le i\le p_3-1$,}\\
		q_i\cdot d & =q_{i+1\bmod p_2}, &\mbox{for~$0\le i\le p_2-1$.}
	\end{align*}
	On the other hand we set for~$p_3\le p_4-1$ and~$p_1\ge 2$,
	\begin{align*}
		q_{2i}\cdot a & =q_{2i+1}, & \mbox{for~$0\le i\le (p_3-3)\div 2$,}\\
		q_{2i+1}\cdot a & =q_{2i+1}, & \mbox{for~$0\le i\le (p_3-3)\div 2$,}\\
		q_i\cdot b & =q_{i-1}, & \mbox{for~$1\le i\le p_3-2$,}\\
		q_{2i-1}\cdot c & =q_{2i}, & \mbox{for~$1\le i\le (p_3-2)\div 2$,}\\
		q_{2i}\cdot c & =q_{2i}, & \mbox{for~$0\le i\le (p_3-2)\div 2$,}\\
		q_{p_3+2i-1}\cdot c & =q_{p_3+2i}, &\mbox{for~$0\le i\le (p_4-p_3-1)/2-1$,}\\
		q_{p_3+2i}\cdot c & =q_{p_3+2i}, &\mbox{for~$0\le i\le (p_4-p_3-1)/2-1$,}\\
		q_i\cdot d & =q_{i+1\bmod p_2}, &\mbox{for~$0\le i\le p_2-1$,}\\
		q_0\cdot e & =q_{p_3-1}, &\\
		q_{p_3+2i-1}\cdot e & =q_{p_3+2i-1}, &\mbox{for~$0\le i\le (p_4-p_3-1)/2-1$,}\\
		q_{p_3+2i}\cdot e & =q_{p_3+2i+1}, &\mbox{for~$0\le i\le (p_4-p_3-1)/2-1$.}
	\end{align*}
	For~$p_3\le p_4-1$ and~$p_1=1$ we elongate the chain of states which are reachable by applying words from~$\{a,c\}^*$ to~$q_0$ by setting
		\begin{align*}
		q_{2i}\cdot a & =q_{2i+1}, & \mbox{for~$0\le i\le (p_3-2)\div 2$}\\
		q_{2i+1}\cdot a & =q_{2i+1}, & \mbox{for~$0\le i\le (p_3-2)\div 2$,}\\
		q_i\cdot b & =q_{i-1}, & \mbox{for~$1\le i\le p_3-1$}\\
		q_{2i-1}\cdot c & =q_{2i}, & \mbox{for~$1\le i\le (p_3-1)\div 2$,}\\
		q_{2i}\cdot c & =q_{2i}, & \mbox{for~$0\le 0\le (p_3-1)\div 2$,}\\
		q_{p_3+2i}\cdot c & =q_{p_3+2i+1}, &\mbox{for~$0\le i\le (p_4-p_3-1)/2-1$,}\\
		q_{p_3+2i-1}\cdot c & =q_{p_3+2i-1}, &\mbox{for~$1\le i\le (p_4-p_3-1)/2-1$,}\\
		q_i\cdot d & =q_{i+1\bmod p_2}, &\mbox{for~$0\le i\le p_2-1$,}\\
		q_0\cdot e & =q_{p_3}, &\\
		q_{p_3+2i}\cdot e & =q_{p_3+2i}, &\mbox{for~$0\le i\le (p_4-p_3-1)/2-1$,}\\
		q_{p_3+2i-1}\cdot e & =q_{p_3+2i}, &\mbox{for~$1\le i\le (p_4-p_3-1)/2-1$.}
	\end{align*}
	
	Additionally to the previously explicitly given transitions we set all other
	transitions to be transitions to the non-accepting sink state~$q_{p_4-1}$ for~$p_3\le p_4-1$ and for~$p_3=p_4$ we set them to be self-loops.
    The automaton~$A$ is depicted in Figure~\ref{fig:overall4control} for the case~$p_3\le p_4-1$, $p_1\ge 2$ (on top), if~$p_3\le p_4-1$, $p_1=1$ (in the middle) and for the case~$p_3=p_4$ (on the bottom).
	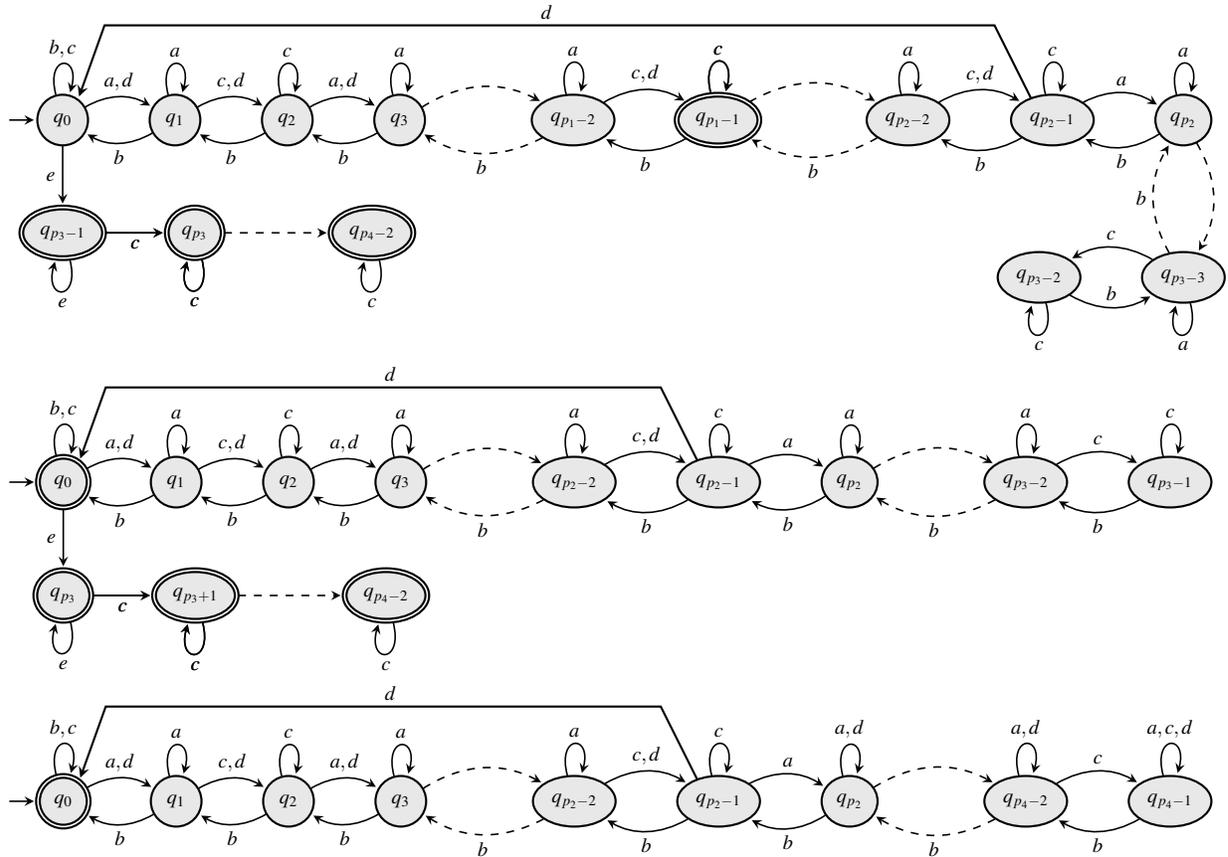
\begin{figure}[!htb]
		\begin{center}
			\begin{tikzpicture}[->,>=stealth,shorten >=1pt,auto,node distance=0.9cm,
	semithick,scale=0.7, every node/.style={scale=0.7},initial text=]
	\tikzstyle{every state}=[fill={rgb:black,1;white,10},ellipse]
	
	\node[initial,state] (p0) [thick]{$q_0$};
	\node[state] (p1) [right = 0.8cm of p0] [thick]{$q_1$};
	\node[state] (p2) [right = 0.8cm of p1] [thick]{$q_2$};
	\node[state] (p3) [right = 0.8cm of p2] [thick]{$q_3$};
	\node[state] (p4) [right = 1.4cm of p3] [thick]{$q_{p_1-2}$};
	\node[state,accepting] (p5) [right = 0.8cm of p4] [thick]{$q_{p_1-1}$};
	
	\node[state] (p6) [right = 1.4cm of p5] [thick]{$q_{p_2-2}$};
	\node[state] (p7) [right = 0.8cm of p6] [thick]{$q_{p_2-1}$};
	\node[state] (p8) [right = 0.8cm of p7] [thick]{$q_{p_2}$};
	
	\node[state] (p9) [below = 1.4cm of p8] [thick]{$q_{p_3-3}$};
	\node[state] (p10) [left = 0.8cm of p9] [thick]{$q_{p_3-2}$};
	
	\node[state,accepting] (p11) [below = 0.8cm of p0] [thick]{$q_{p_3-1}$};
	\node[state,accepting] (p12) [right = 0.8cm of p11] [thick]{$q_{p_3}$};
	\node[state,accepting] (p13) [right = 1.4cm of p12] [thick]{$q_{p_4-2}$};
	
	\path

	(p0) edge[bend left] node[midway,above]{$a,d$} (p1)
	(p0) edge[loop above] node[midway]{$b,c$} (p0)
	
	(p1) edge[bend left] node[midway,above]{$c,d$} (p2)
	(p1) edge[loop above] node[midway]{$a$} (p1)
	(p1) edge[bend left] node[midway]{$b$} (p0)
	
	(p2) edge[bend left] node[midway,above]{$a,d$} (p3)
	(p2) edge[loop above] node[midway]{$c$} (p2)
	(p2) edge[bend left] node[midway]{$b$} (p1)
	
	(p3) edge[bend left,dashed] node[midway,above]{} (p4)
	(p3) edge[loop above] node[midway]{$a$} (p3)
	(p3) edge[bend left] node[midway]{$b$} (p2)
	
	(p4) edge[bend left] node[midway,above]{$c,d$} (p5)
	(p4) edge[loop above] node[midway]{$a$} (p4)
	(p4) edge[bend left,dashed] node[midway]{$b$} (p3)

	(p5) edge[loop above] node[midway]{$c$} (p4)
	(p5) edge[bend left] node[midway]{$b$} (p4)
	(p5) edge[bend left,dashed] node[midway]{} (p6)
	
	(p6) edge[bend left] node[midway,above]{$c,d$} (p7)
	(p6) edge[loop above] node[midway]{$a$} (p6)
	(p6) edge[bend left,dashed] node[midway]{$b$} (p5)
	
	(p7) edge[bend left] node[midway,above]{$a$} (p8)
	(p7) edge[loop above] node[midway]{$c$} (p7)
	(p7) edge[bend left] node[midway]{$b$} (p6)
	
	(p8) edge[loop above] node[midway]{$a$} (p8)
	(p8) edge[bend left] node[midway]{$b$} (p7)
	(p8) edge[bend left,dashed] node[midway,right]{} (p9)
	
	(p9) edge[loop below] node[midway]{$a$} (p9)
	(p9) edge[bend left,dashed] node[midway]{$b$} (p8)
	(p9) edge[bend right] node[midway,above]{$c$} (p10)
	
	(p10) edge[loop below] node[midway]{$c$} (p10)
	(p10) edge[bend right] node[midway]{$b$} (p9)

	(p0) edge[] node[midway,left]{$e$} (p11)
	(p11) edge[loop below] node[midway]{$e$} (p11)
	(p11) edge[] node[midway, below]{$c$} (p12)
	(p12) edge[loop below] node[midway]{$c$} (p12)
	
	(p11) edge[] node[midway, below]{$c$} (p12)
	(p12) edge[loop below] node[midway]{$c$} (p12)
	
	(p12) edge[dashed] node[midway]{} (p13)
	(p5) edge[loop above] node[midway]{$c$} (p4)
	(p13) edge[loop below] node[midway]{$c$} (p13)

%	(p7) edge[bend right=20] node[midway,above]{$d$} (p0)
	
	;
	\draw[thick,sloped] ($(p7.west)+(0.4,0.4)$)--($(p7.west)+(-0.3,1.8)$)
	--($(p0.east)+(0.3,1.8)$)--($(p0.east)+(-0.2,0.4)$)
	;
	
	\node[] (d) at ($(p4.west)+(0.3,2.06)$) [minimum size=0.1cm] [thick]{$d$};		
	
	%\draw[thick] ($(p7.west)+(0.4,0.4)$) .. controls +(up:3cm) and +(right:-10.5cm) .. node[above,sloped] {label} ($(p0.east)+(-0.2,0.4)$);

\end{tikzpicture}	
			\begin{tikzpicture}[->,>=stealth,shorten >=1pt,auto,node distance=0.9cm,
	semithick,scale=0.7, every node/.style={scale=0.7},initial text=]
	\tikzstyle{every state}=[fill={rgb:black,1;white,10},ellipse]
	
	\node[initial,state,accepting] (p0) [thick]{$q_0$};
	\node[state] (p1) [right = 0.8cm of p0] [thick]{$q_1$};
	\node[state] (p2) [right = 0.8cm of p1] [thick]{$q_2$};
	\node[state] (p3) [right = 0.8cm of p2] [thick]{$q_3$};
	%\node[state] (p4) [right = 1.4cm of p3] [thick]{$q_{p_1-2}$};
	%\node[state] (p5) [right = 0.8cm of p4] [thick]{$q_{p_1-1}$};
	
	\node[state] (p6) [right = 1.4cm of p3] [thick]{$q_{p_2-2}$};
	\node[state] (p7) [right = 0.8cm of p6] [thick]{$q_{p_2-1}$};
	\node[state] (p8) [right = 0.8cm of p7] [thick]{$q_{p_2}$};
	
	\node[state] (p9) [right = 1.4cm of p8] [thick]{$q_{p_3-2}$};
	\node[state] (p10) [right = 0.8cm of p9] [thick]{$q_{p_3-1}$};
	
	\node[state,accepting] (p11) [below = 0.8cm of p0] [thick]{$q_{p_3}$};
	\node[state,accepting] (p12) [right = 0.8cm of p11] [thick]{$q_{p_3+1}$};
	\node[state,accepting] (p13) [right = 1.4cm of p12] [thick]{$q_{p_4-2}$};
	
	\path

	(p0) edge[bend left] node[midway,above]{$a,d$} (p1)
	(p0) edge[loop above] node[midway]{$b,c$} (p0)
	
	(p1) edge[bend left] node[midway,above]{$c,d$} (p2)
	(p1) edge[loop above] node[midway]{$a$} (p1)
	(p1) edge[bend left] node[midway]{$b$} (p0)
	
	(p2) edge[bend left] node[midway,above]{$a,d$} (p3)
	(p2) edge[loop above] node[midway]{$c$} (p2)
	(p2) edge[bend left] node[midway]{$b$} (p1)
	
	(p3) edge[bend left,dashed] node[midway,above]{} (p6)
	(p3) edge[loop above] node[midway]{$a$} (p3)
	(p3) edge[bend left] node[midway]{$b$} (p2)
	
%	(p4) edge[bend left] node[midway,above]{$c,d$} (p5)
%	(p4) edge[loop above] node[midway]{$a$} (p4)
%	(p4) edge[bend left,dashed] node[midway]{$b$} (p3)
%	
%	
%	(p5) edge[loop above] node[midway]{$c$} (p4)
%	(p5) edge[bend left] node[midway]{$b$} (p4)
%	(p5) edge[bend left,dashed] node[midway]{} (p6)
	
	(p6) edge[bend left] node[midway,above]{$c,d$} (p7)
	(p6) edge[loop above] node[midway]{$a$} (p6)
	(p6) edge[bend left,dashed] node[midway]{$b$} (p3)
	
	(p7) edge[bend left] node[midway,above]{$a$} (p8)
	(p7) edge[loop above] node[midway]{$c$} (p7)
	(p7) edge[bend left] node[midway]{$b$} (p6)
	
	(p8) edge[loop above] node[midway]{$a$} (p8)
	(p8) edge[bend left] node[midway]{$b$} (p7)
	(p8) edge[bend left,dashed] node[midway,right]{} (p9)
	
	(p9) edge[loop above] node[midway]{$a$} (p9)
	(p9) edge[bend left,dashed] node[midway]{$b$} (p8)
	(p9) edge[bend left] node[midway,above]{$c$} (p10)
	
	(p10) edge[loop above] node[midway]{$c$} (p10)
	(p10) edge[bend left] node[midway]{$b$} (p9)

	(p0) edge[] node[midway,left]{$e$} (p11)
	(p11) edge[loop below] node[midway]{$e$} (p11)
	(p11) edge[] node[midway, below]{$c$} (p12)
	(p12) edge[loop below] node[midway]{$c$} (p12)
	
	(p11) edge[] node[midway, below]{$c$} (p12)
	(p12) edge[loop below] node[midway]{$c$} (p12)
	
	(p12) edge[dashed] node[midway]{} (p13)
%	(p5) edge[loop above] node[midway]{$c$} (p4)
	(p13) edge[loop below] node[midway]{$c$} (p13)

%	(p7) edge[bend right=20] node[midway,above]{$d$} (p0)
	
	;
	\draw[thick,sloped] ($(p7.west)+(0.4,0.4)$)--($(p7.west)+(-0.3,1.8)$)
	--($(p0.east)+(0.3,1.8)$)--($(p0.east)+(-0.2,0.4)$)
	;
	
	\node[] (d) at ($(p3.west)+(0.3,2.06)$) [minimum size=0.1cm] [thick]{$d$};		
	
	%\draw[thick] ($(p7.west)+(0.4,0.4)$) .. controls +(up:3cm) and +(right:-10.5cm) .. node[above,sloped] {label} ($(p0.east)+(-0.2,0.4)$);

\end{tikzpicture}	
			\begin{tikzpicture}[->,>=stealth,shorten >=1pt,auto,node distance=0.9cm,
	semithick,scale=0.7, every node/.style={scale=0.7},initial text=]
	\tikzstyle{every state}=[fill={rgb:black,1;white,10},ellipse]
	
	\node[initial,state,accepting] (p0) [thick]{$q_0$};
	\node[state] (p1) [right = 0.8cm of p0] [thick]{$q_1$};
	\node[state] (p2) [right = 0.8cm of p1] [thick]{$q_2$};
	\node[state] (p3) [right = 0.8cm of p2] [thick]{$q_3$};
	%\node[state] (p4) [right = 1.4cm of p3] [thick]{$q_{p_1-2}$};
	%\node[state] (p5) [right = 0.8cm of p4] [thick]{$q_{p_1-1}$};
	
	\node[state] (p6) [right = 1.4cm of p3] [thick]{$q_{p_2-2}$};
	\node[state] (p7) [right = 0.8cm of p6] [thick]{$q_{p_2-1}$};
	\node[state] (p8) [right = 0.8cm of p7] [thick]{$q_{p_2}$};
	
	\node[state] (p9) [right = 1.4cm of p8] [thick]{$q_{p_4-2}$};
	\node[state] (p10) [right = 0.8cm of p9] [thick]{$q_{p_4-1}$};
		
	\path

	(p0) edge[bend left] node[midway,above]{$a,d$} (p1)
	(p0) edge[loop above] node[midway]{$b,c$} (p0)
	
	(p1) edge[bend left] node[midway,above]{$c,d$} (p2)
	(p1) edge[loop above] node[midway]{$a$} (p1)
	(p1) edge[bend left] node[midway]{$b$} (p0)
	
	(p2) edge[bend left] node[midway,above]{$a,d$} (p3)
	(p2) edge[loop above] node[midway]{$c$} (p2)
	(p2) edge[bend left] node[midway]{$b$} (p1)
	
	(p3) edge[bend left,dashed] node[midway,above]{} (p6)
	(p3) edge[loop above] node[midway]{$a$} (p3)
	(p3) edge[bend left] node[midway]{$b$} (p2)
	
%	(p4) edge[bend left] node[midway,above]{$c,d$} (p5)
%	(p4) edge[loop above] node[midway]{$a$} (p4)
%	(p4) edge[bend left,dashed] node[midway]{$b$} (p3)
%	
%	
%	(p5) edge[loop above] node[midway]{$c$} (p4)
%	(p5) edge[bend left] node[midway]{$b$} (p4)
%	(p5) edge[bend left,dashed] node[midway]{} (p6)
	
	(p6) edge[bend left] node[midway,above]{$c,d$} (p7)
	(p6) edge[loop above] node[midway]{$a$} (p6)
	(p6) edge[bend left,dashed] node[midway]{$b$} (p3)
	
	(p7) edge[bend left] node[midway,above]{$a$} (p8)
	(p7) edge[loop above] node[midway]{$c$} (p7)
	(p7) edge[bend left] node[midway]{$b$} (p6)
	
	(p8) edge[loop above] node[midway]{$a,d$} (p8)
	(p8) edge[bend left] node[midway]{$b$} (p7)
	(p8) edge[bend left,dashed] node[midway,right]{} (p9)
	
	(p9) edge[loop above] node[midway]{$a,d$} (p9)
	(p9) edge[bend left,dashed] node[midway]{$b$} (p8)
	(p9) edge[bend left] node[midway,above]{$c$} (p10)
	
	(p10) edge[loop above] node[midway]{$a,c,d$} (p10)
	(p10) edge[bend left] node[midway]{$b$} (p9)

%	(p7) edge[bend right=20] node[midway,above]{$d$} (p0)
	
	;
	\draw[thick,sloped] ($(p7.west)+(0.4,0.4)$)--($(p7.west)+(-0.3,1.8)$)
	--($(p0.east)+(0.3,1.8)$)--($(p0.east)+(-0.2,0.4)$)
	;
	
	\node[] (d) at ($(p3.west)+(0.3,2.06)$) [minimum size=0.1cm] [thick]{$d$};		
	
	%\draw[thick] ($(p7.west)+(0.4,0.4)$) .. controls +(up:3cm) and +(right:-10.5cm) .. node[above,sloped] {label} ($(p0.east)+(-0.2,0.4)$);

\end{tikzpicture}	
		\end{center}
		\caption{The automaton~$A$ for the case~$p_3\le p_4-1$, $p_1\ge 2$ (on top), if~$p_3\le p_4-1$, $p_1=1$ (in the middle) and for the case~$p_3=p_4$ (on the bottom). For the first two cases the state~$q_{p_4-1}$ is a non-accepting sink state and all not shown transitions are mappings onto~$q_{p_4-1}$. 
		In the case~$p_3=p_4$ the letter~$e$ is not needed. 
		Recall, that the~$a$-, $c$-, and~$e$-transitions in all cases depend on the parity of~$p_3-2$ and~$p_4-2$, \resp.}
		\label{fig:overall4control}
	\end{figure}
	We will use small claims for making it easier to prove that the language~$L:=L(A)$ fulfills
	the requested properties.

	\begin{clm}
		The automaton~$A$ is minimal.
	\end{clm}
	
	\begin{proof}
		We observe that for all states in~$S_1:=\{q_0,$ $q_1$, $\dots,$ $q_{p_1-2}\}$ there is a unique shortest word in~$\{a,c\}^*$ mapping the state onto~$q_{p_1-1}$. 
		The analogue is true for the states in~$S_2:=\{q_{p_1-1},$ $q_{p_1}$, $\dots,$ $q_{p_3-2}\}$
		and the set~$\{b\}^*$.
		Therefore the above mentioned states cannot contain a pair of equivalent states.
		Additionally for all states~$S_3:=\{q_{p_3},$ $q_{p_3+1}$, $\dots,$ $q_{p_4-2},q_0\}$ there
		is a unique shortest word in~$\{c,e\}^*$ mapping the state onto the state~$q_{p_4-2}$ which implies~$S_3$ cannot contain equivalent states.
		Since~$S_1\cdot b^{p_1}=\{q_0\}$ and~$S_3\cdot b^{p_1}=\{q_{p_4-1}\}$ we obtain
		that there are no states in~$S_1\cup S_2\cup S_3\cup \{q_{p_4-1}\}$ which are
		equivalent. Indeed this directly implies that~$A$ is minimal.
	\end{proof}

	\begin{clm}
		We have~$\mpl(L)=p_2$.
	\end{clm}
	
	\begin{proof}
		Due to the fact that~$L\cap \{d\}^*=(\{d\}^{p_2})^*\{d\}^{p_1-1}$ we have that
		the word~$d^{p_2+p_1-1}$ is only pumpable by the sub-word~$d^{p_2}$ and no shorter sub-word. Indeed this implies that~$\mpl(L)\ge p_2$.
		We will show that each word~$\tilde{w}\in L$ of length at least~$p_1$ is pumpable by a sub-word
		of its~$p_2$-prefix.
		Therefore we distinguish between the several beginnings of~$\tilde{w}$:
		\begin{itemize}
			\item The first letter of~$\tilde{w}$ is an~$a$ or a~$d$. 
			Here we observe that either~$\tilde{w}$ contains one of the words~$ab$, $cb$, $db$, $aa$ or $cc$ in its~$p_2$-prefix or its~$p_2$-prefix~$w_1$ is from~$\{a,c,d\}^{p_2}$ such that~$q_0\cdot w_1=q_0$.
			
			If~$\tilde{w}$ contains one of the words~$ab$, $cb$, $db$, $aa$ or $cc$ in its~$p_2$-prefix then $\tilde{w}$ can be pumped by the sub-words~$ab$, $cb$, $db$, $a$ and $c$, \resp.
			
			If~$\tilde{w}$ has a~$p_2$-prefix~$w_1$ which is from~$\{a,b,c\}^{p_2}$ such that~$q_0\cdot w_1=q_0$
			then we can pump~$\tilde{w}$ by~$w_1$ since~$q_0\cdot w_1^i=q_0$ for all~$i\ge 0$.
			
			\item The word~$\tilde{w}$ starts with the letter~$b$ or~$c$.
			It is obvious that~$\tilde{w}$ is pumpable by its first letter.
			
			\item If the word~$\tilde{w}$ has~$e$ as its first letter we observe that~$\tilde{w}\in\{e,c\}^*$.
			For~$p_2=1$ we can pump~$\tilde{w}$ by its first letter since~$q_0\cdot c=q_0$ and~$q_0\cdot e^i=q_{p_3-1}$ for all~$i\ge 1$, which are both accepting states.
			For~$p_2\ge 2$ we can pump~$\tilde{w}$ by its second letter since~$q_{p_3}\cdot c^i=q_{p_3}$ and~$q_{p_3-1}\cdot e^i=q_{p_3-1}$ for
			all~$i\ge 0$.
		\end{itemize}

	\end{proof}

	\begin{clm}
		We have~$\mpc(L)=p_1$.
	\end{clm}
	
	\begin{proof}
		Since we have shown that each word of length at least~$p_1$ is pumpable by its~$p_2-$prefix it remains to observe that the word~$d^{p_1-1}$ is not pumpable since~$L\cap \{d\}^*=(\{d\}^{p_2})^*\{d\}^{p_1-1}$.
	\end{proof}

	\begin{clm}
		We have~$\mps(L)=p_3$.
	\end{clm}

	\begin{proof}
		Observe that for~$p_1=1$ and~$p_3\le p_4-1$ the chain of non-sink states which are reachable from
		the initial state in~$A$ by applying a word in~$\{a,c\}^*$ is exactly one state
		longer as for~$p_1\ge 2$ and~$p_3\le p_4-1$.
		Therefore we have that~$\tilde{w}=(ac)^{(p_3-2\div 2)}a^{p_3-2\bmod 2}b^{p_3-2}e$ is not pumpable by any sub-word of~$w=b^{p_3-2}e$ for~$p_1\ge 2$
		and~$\tilde{w}=(ac)^{(p_3-1\div 2)}a^{p_3-1\bmod 2}b^{p_3-1}$ is not pumpable
		by any sub-word of~$w=b^{p_3-1}$ for~$p_1=1$ which implies
		that~$\mps(L)\ge p_3$.
		We now distinguish between all possible words~$w\in\Sigma^*$ with~$|w|=p_3$
		and the words~$\tilde{w}$ which can contain them to give a sub-word~$y$ of~$w$
		such that~$\tilde{w}$ is pumpable by~$y$:
		\begin{itemize}
			\item If~$w$ contains~$aa$ or $cc$ then~$\tilde{w}$ can be pumped by~$y=a$ and~$c$, \resp.
			
			\item In the case~$w$ contains a sub-word in~$\{\,xb\mid x\in\{a,c,d\}\,\}$ then we can pump~$\tilde{w}$
			by~$y=x$ if~$x$ induces a self-loop for the according state
			or by~$y=b$ if~$b$ from~$xb$ induces a self-loop on the according state
			or by~$xb$ otherwise.
			The last way of pumping is possible since~$xb$ induces a self-loop on the according state.
			
			\item The case that~$w$ contains a sub-word from~$\{bx\mid x\in\{a,c,d\}\}$
			can be treated similarly as above.
			
			\item If~$w$ contains a sub-word~$y$ from~$\{a,c,d\}^*$ with length~$p_2$
			such that~$\tilde{w}=uxyzv$ and~$w=xyz$ for words~$u,x,z,v\in\Sigma^*$,
			$q_0\cdot ux\in\{q_0,q_1,\dots,q_{p_2}\}$ and~$q_0\cdot uxy=q_0\cdot ux$.
			Clearly~$\tilde{w}$ can be pumped by~$y$.
			
			\item The word~$w$ contains the letter~$b=y$ such that~$\tilde{w}=uxyzv$ and~$w=xyz$ for words~$u,x,z,v\in\Sigma^*$, and~$q_0\cdot ux=q_0$.
			Then~$\tilde{w}$ can be pumped by~$y=b$ because~$q_0\cdot uxy=q_0\cdot y^i=q_0\cdot b^i=q_0$ for all~$i\ge 0$.
			
			\item If the word~$w$ contains the sub-word~$ec$ or~$ee$ then we can pump~$\tilde{w}$ by~$y=c$ or~$y=e$, \resp.
			
		\end{itemize}
		It remains to observe that~$w$ has to contain one of the previously mentioned
		sub-words.
		Therefore we study how long the longest prefix~$w'$ of~$w$ in~$\{a,b,c,d\}^*$ can be such that none of the above-mentioned sub-words are contained.
		Afterwards we elongate this prefix by a word in~$\Sigma^*$.

		First one may understand that for any given state~$q$ of~$A$ the longest word~$w'$ in~$\{a,b,c,d\}^*$, that cannot be decomposed into~$w'=xyz$ for  words~$x$, $y$, $z\in\Sigma^*$ such that~$|y|\ge 1$ and~$q\cdot x=q\cdot xy$, has length at most~$p_3-1$ for~$p_3=p_4$ and length at most~$p_3-2$ for~$p_3\le p_4-1$.
		Roughly speaking this can be seen by observing that the longest such word has to map the state~$q$ onto each of the states~$q_0,q_1,\dots,q_{p_3-1}$ for~$p_3=p_4$ and onto each of the states~$q_0,q_1,\dots,q_{p_3-2}$ for~$p_3\le p_4-1$.
		The only possibilities to elongate such a word~$w'$ are to either violate the 
		previously described decomposing property or to elongate~$w'$ by the letter~$e$.
		Due to the construction of the automaton the word~$w'e$ can only be a sub-word 
		of a word~$\tilde{w}\in L$ iff~$\tilde{w}=uw'ew''v$ for~$q_0\cdot uw'e=q_{p_3}$,~$w'',v\in\Sigma^*$, and~$w=w'ew''$.
		Again the transition mapping of~$A$ implies that~$w''$ is empty or starts with one of the letters~$c$ and $e$.
		Indeed this implies that~$w=w'ew''$ either has length~$|w|=|w'e|\le p_3-1$ for~$w''=\lambda$ or contains one of the sub-words~$ee$ or~$ec$ and is therefore
		pumpable by its~$p_3$-th letter.
		
		One observes that if we choose~$w'$ to be not maximal it similarly 
		that~$w$ either has length less than~$p_3$ or it contains one of the sub-words~$y$ mentioned above such that that~$\tilde{w}$ is pumpable by~$y$.
	\end{proof}
	In conclusion we have that~$\mpc(L)=p_1$,~$\mpl(L)=p_2$,~$\mps(L)=p_3$, and~$\dsc(L)=p_4$ for~$p_3\le p_4-1$.
	Due to Theorem~\ref{thm:loops-only-decrease-mpx} we directly obtain for~$p_3=p_4$ that the according pumping constants have to be at most 
	equal to the pumping constants in the case~$p_3\le p_4-1$.
	In turn we observe that the witnesses for~$\mpc(L)\ge p_1$ and~$\mpl(L)\ge p_2$
	can also applied for~$p_3=p_4$.
	Additionally the word~$\tilde{w}=(ac)^{(p_3-1\div 2)}a^{p_3-1\bmod 2}b^{p_3-1}$
	with~$w=b^{p_3-1}$ witnesses~$\mps(L)\ge p_3$ for~$p_3=p_4$. The minimality of~$A$ can be shown similarly as for~$p_3\le p_4-1$.
	Therefore we conclude that~$\mpc(L)=p_1$,~$\mpl(L)=p_2$,~$\mps(L)=p_3$, and~$\dsc(L)=p_4$.
\end{proof}

\subsection{Operational Complexity of Sub-Word Pumping}

We study the effect of regularity preserving standard formal language
operations on the minimal pumping constant w.r.t.\
Lemma~\ref{lem:pumping-xy-length} and compare them to previously
obtained results~\cite{DaJe22} for the other minimal pumping
constants. To this end we need some notation: let~$\circ$ be a
regularity preserving $n$-ary function on languages and
$K\in\{\mpc,\mpl,\mps\}$. Then, we define
$g_\circ^K(k_1 , k_2 , \ldots, k_n)$ as the set of all numbers~$k$
such that there are regular languages $L_1,L_2,\ldots,L_n$ with
$K(L_i)=k_i$, for $1\leq i\leq n$ and
$K(\circ(L_1, L_2,\ldots, L_n)) = k$.  Results for some regularity
preserving operations on~$\mpc$ and~$\mpl$ can be found in the
comprehensive Table~\ref{tab:results}. The set of all natural numbers
not including zero is denoted by~$\mathbb{N}$; if zero is included,
then we write~$\mathbb{N}_0$ instead. The gray shaded entries in
Table~\ref{tab:results} are new results, left open results, or corrected
results from~\cite{DaJe22}. We only give
the proofs for two of these new results, namely Kleene star and
intersection.

Let us start with the Kleene star operation. In~\cite{DaJe22} it was
shown that for the Kleene star operation the following results hold:
$$
g^{\mpc}_{*}(n)=\{1\}
\quad\mbox{and}\quad 
g^{\mpl}_{*}(n)=
\begin{cases}
	\{1\}, & \mbox{if $n=0$,}\\
	\{1,2,\dots,n\}, & \mbox{otherwise,} 
\end{cases}
$$
for every~$n\ge 0$. 
\newcommand{\mr}[2]{\multirow{#1}{*}{#2}}%
\newcommand{\ccg}{\cellcolor[gray]{0.9}}%
\begin{sidewaystable}
	\begin{center}
		\centering
		\resizebox{\textheight}{!}{%
                \begin{tabular}{lccc}\toprule
                  & \multicolumn{3}{c}{Minimal pumping constant}\\\cmidrule(lr){2-4} 
    Operation  & $\mpc$ & $\mpl$ & $\mps$\\\cmidrule(lr){1-4}
    Kleene star & $\{1\}$ & 
                                            \begin{tabular}{ll}
                                             $\{1\},$ & \mbox{if $n=0$,}\\ 
                                             $\{1,2,\dots,n\},$ & \mbox{otherwise.}
                                            \end{tabular}

    & \ccg
      \begin{tabular}{ll}
        $\{1\},$ & \mbox{if $n=0$,}\\
        $\{1,2,\dots,2n-1\},$ & \mbox{otherwise.}
      \end{tabular}
    \\\cmidrule(lr){1-4}
	Reversal &$\{n\}$ & \begin{tabular}{ll}
		$\{0\},$ & \mbox{if $n=0$,}\\ 
		$\N,$ & \mbox{otherwise.}
	\end{tabular}
 	& $\{n\}$\ccg\\\cmidrule(lr){1-4}
 	Complement &\begin{tabular}{ll}
 	$\{1\},$ & \mbox{if $n=0$,}\\
 	$\N_0\setminus\{1\},$ & \mbox{if $n=1$,}\\
 	$\N,$ & \mbox{otherwise.} 
 		\end{tabular}
 	&
 	\begin{tabular}{ll}
 		$\{1\},$ & \mbox{if $n=0$}\\
 		$\N_0\setminus\{1\},$ & \mbox{if $n=1$,}\\
 		$\N,$ & \mbox{otherwise.} 
 	\end{tabular}
 	
 	 & \begin{tabular}{ll}
 	 	$\{1\},$ & \mbox{if $n=0$,}\\
 	 	$\N_0\setminus\{1\},$ & \mbox{if $n=1$,}\\
 	 	$\N,$ & \mbox{otherwise.} 
 	 \end{tabular}
  	\ccg\\\cmidrule(lr){1-4}
    Prefix-Closure &   \begin{tabular}{ll}
			    	$\{0\},$ & \mbox{if $n=0$,}\\
			    	$\mathbb{N},$ & \mbox{otherwise.}          
			    \end{tabular}
 			 &
                 \begin{tabular}{ll}
                   $\{0\},$ & \mbox{if $n=0$,}\\
                   $\{1,2,\dots,n\},$ & \mbox{otherwise.}          
                 \end{tabular}
             & \ccg
	             \begin{tabular}{ll}
	             	$\{0\},$ & \mbox{if $n=0$,}\\
	             	$\{1,2,\dots,n\},$ & \mbox{otherwise.}          
	             \end{tabular}

              \\\cmidrule(lr){1-4}
    Suffix-Closure &   \begin{tabular}{ll}
    	$\{0\},$ & \mbox{if $n=0$,}\\
    	$\mathbb{N},$ & \mbox{otherwise.}          
    \end{tabular}
    &
    \begin{tabular}{ll}
    	$\{0\},$ & \mbox{if $n=0$,}\\
    	$\{1\},$ & \mbox{if $n=1$,}\\
    	$\N,$ & \mbox{otherwise.}          
    \end{tabular}
    & \ccg
    \begin{tabular}{ll}
    	$\{0\},$ & \mbox{if $n=0$,}\\
    	$\{1,2,\dots,n\},$ & \mbox{otherwise.}          
    \end{tabular}
    
    \\\cmidrule(lr){1-4}
    Union &   \begin{tabular}{ll}
    	$\max\{m,n\},$ & \mbox{if $m=0$ or~$n=0$,}\\
    	$\{1,2,\dots,\max\{m,n\}\},$ & \mbox{otherwise.}          
    \end{tabular}
    &
    \begin{tabular}{ll}
    	$\max\{m,n\},$ & \mbox{if $m=0$ or~$n=0$,}\\
    	$\{1,2,\dots,\max\{m,n\}\},$ & \mbox{otherwise.}          
    \end{tabular}
    & \ccg
    \begin{tabular}{ll}
    	$\max\{m,n\},$ & \mbox{if $m=0$ or~$n=0$,}\\
    	$\{1,2,\dots,\max\{m,n\}\},$ & \mbox{otherwise.}          
    \end{tabular}

	\\\cmidrule(lr){1-4}
    Set-Subtraction &
    \begin{tabular}{ll}
    	$\{0\},$ & \mbox{if $m= 0,n\ge 0$,}\\
    	$\{m\},$ & \mbox{if $m\ge 0,n= 0$,}\\
    	\ccg $\N_0\setminus\{1\},$ & \ccg \mbox{if $m\ge 1,n= 1$,}\\
    	$\N_0,$ & \mbox{otherwise.}          
    \end{tabular}
    &
     \begin{tabular}{ll}
    	$\{0\},$ & \mbox{if $m= 0,n\ge 0$,}\\
    	$\{m\},$ & \mbox{if $m\ge 0,n= 0$,}\\
    	\ccg$\N_0\setminus\{1\},$ & \ccg\mbox{if $m\ge 1,n= 1$,}\\
    	$\N_0,$ & \mbox{otherwise.}          
    \end{tabular}

    &  \ccg  \begin{tabular}{ll}
    	$\{0\},$ & \mbox{if $m= 0,n\ge 0$,}\\
    	$\{m\},$ & \mbox{if $m\ge 0,n= 0$,}\\
    	$\N_0\setminus\{1\},$ & \ccg\mbox{if $m\ge 1,n= 1$,}\\
    	$\N_0,$ & \mbox{otherwise.}          
    \end{tabular}\\
    \cmidrule(lr){1-4}
    Concatenation &
    	\begin{tabular}{ll}
    		$\{0\}$, & \mbox{if $m=0$ or~$n= 0$,}\\
    		$\{1,2,\ldots,m+n-1\}$, & \mbox{otherwise.}    
    	\end{tabular}
    &
    \begin{tabular}{ll}
    	$\{0\}$, & \mbox{if $m=0$ or~$n= 0$,}\\
    	$\{1,2,\ldots,m+n-1\}$, & \mbox{otherwise.}    
    \end{tabular}
    &	\ccg
    \begin{tabular}{ll}
    	$\{0\}$, & \mbox{if $m=0$ or~$n= 0$,}\\
    	$\{1,2,\ldots,m+n-1\}$, & \mbox{otherwise.}    
    \end{tabular}\\
    \cmidrule(lr){1-4}
    Intersection &
    \begin{tabular}{ll}
    $\{0\}$, & \mbox{if $m=0$ or $n= 0$,}\\
    $\N_0\setminus\{2\}$, & \mbox{if $m=n=1$,}\\
    $\N_0$, & \mbox{otherwise.}
	\end{tabular}
    
     & \ccg
     \begin{tabular}{ll}
     $\{0\}$, & \mbox{if $m=0$ or $n= 0$,}\\
     $\{1\}$, & \mbox{if $m=n=1$,}\\
     $\N_0$, & \mbox{otherwise.}
    \end{tabular}
     &\ccg
     \begin{tabular}{ll}
     	$\{0\}$, & \mbox{if $m=0$ or $n= 0$,}\\
      	$\{1\}$, & \mbox{if $m=n=1$,}\\
      	$\N_0$, & \mbox{otherwise.}
     \end{tabular}
      \\
       \cmidrule(lr){1-4}
      Symmetric Difference &
      \begin{tabular}{ll}
      	$\max\{m,n\}$, & \mbox{if $m=0$ or~$n=0$,}\\
      	$\N_0\setminus\{1\}$, & \mbox{if $m=n=1$,}\\
      	$\N_0$, & \mbox{if $m=n>1$,}\\
      	$\N$, & \mbox{otherwise.}  
      \end{tabular}
      
      & 
      \begin{tabular}{ll}
      	$\max\{m,n\}$, & \mbox{if $m=0$ or~$n=0$,}\\
      	\ccg $\N_0\setminus\{1\}$, & \ccg\mbox{if $m=n=1$,}\\
      	\ccg $\N_0$, & \ccg\mbox{if $m=n>1$,}\\
      	\ccg $\N$, & \ccg\mbox{otherwise.}  
      \end{tabular}
      &\ccg
      
      \begin{tabular}{ll}
      	$\max\{m,n\}$, & \mbox{if $m=0$ or~$n=0$,}\\
      	\ccg $\N_0\setminus\{1\}$, & \ccg\mbox{if $m=n=1$,}\\
      	\ccg $\N_0$, & \ccg\mbox{if $m=n>1$,}\\
      	\ccg $\N$, & \ccg\mbox{otherwise.}  
      \end{tabular}      
      \\

    \bottomrule
  \end{tabular}}
  \caption{Results on the operational complexity of the minimal pumping constants~$\mpc$,\ $\mpl$, and $\mps$. The results for the former two minimal pumping constants are from~\cite{DaJe22}. Gray shaded entries indicate new results, previous left open results, or corrected ones. Here~$\mathbb{N}$ refers to the set of all natural number \emph{not} including zero; if zero is included we refer to this set as~$\mathbb{N}_0$.}
  \label{tab:results}

\end{center}
\end{sidewaystable}
For the minimal pumping constant~\mps\ a larger set of numbers is attainable
as we show next.

\begin{thm}\label{thm:mps-star}
	It holds
	$$g^{\mps}_{*}(n)=
	\begin{cases}
		\{1\}, & \mbox{if $n=0$,}\\
		\{1,2,\dots,2n-1\}, & \mbox{otherwise.}
	\end{cases}
	$$
\end{thm}

\begin{proof}
	First we look at the case where~$n=0$. Afterwards we argue why, for~$n\ge 1$, no value in~$\N_0\setminus \{1,2,\dots,2n-1\}$ can be reached, and at last we define languages~$L_{n,k}$ with the property that $\mps(L_{n,k})=n$ and for the Kleene star of~$L_{n,k}$ we have $\mps(L_{n,k}^*)=k$. 
	
	For~$\mps(L_{0,k})=n=0$ we observe that~$L_{0,k}=\emptyset$.
	So we have that~$\mps(L_{0,k})=n=0$ implies that~$k=\mps(L_{0,k}^*)=\mps(\emptyset^*)=\mps(\{\lambda\})=1$.
	Next we show that for any language~$L$ with~$\mps(L)=n$ we have that~$\mps(L^*)\le 2n-1$.
	We observe that each non-empty word~$\tilde{w}\in L^*$ is equal to~$\tilde{w}_1\tilde{w}_2\dots \tilde{w}_t$, for~$\tilde{w}_1,\tilde{w}_2,\dots, \tilde{w}_t\in L$. We know for each of those words that each of their
	sub-words of length~$n$ can be pumped by a sub-word of length at most~$n$.
	Assume that~$\mps(L^*)\ge 2n$ and the sub-word~$w$ of~$\tilde{w}$ is a witness for that, which means there are words~$u$ and~$v$ in~$\Sigma^*$ such that~$\tilde{w}=uwv\in L^*$
	cannot be pumped by a sub-word of the~$2n-1$-prefix of~$w$.
	W.l.o.g.\ we assume that the words~$\tilde{w}_1,\tilde{w}_2,\dots, \tilde{w}_t\in L$ are not empty. Obviously, we have that $$\tilde{w}=\tilde{w}_1\tilde{w}_2\dots \tilde{w}_t=\tilde{w}_1\tilde{w}_2\dots \tilde{w}_{i-1}w'_iww'_j\tilde{w}_{j+1}\dots \tilde{w}_t$$ 
	for~$w'_iww'_j=\tilde{w}_i\tilde{w}_{i+1}\dots \tilde{w}_{j-1}\tilde{w}_j$.
	We know that each sub-word of length~$n$ of~$\tilde{w}_i$ and~$\tilde{w}_{i+1}$ can be pumped by one of its sub-words.
	Especially this holds for the~$n$-suffix of~$\tilde{w}_i$. If this suffix is contained in~$w$ than~$uwv=\tilde{w}$ can be pumped by that sub-word of~$w$ which contradicts the assumption that~$w$ is a witness for~$\mps(L^*)\ge 2n$.
	The analogue holds true if the~$n$-prefix of~$\tilde{w}_{i+1}$ is contained in~$w$.
	Additionally, if~$\tilde{w}_i$ (or $\tilde{w}_{i+1}$, \resp) is completely contained in~$w$ and has length less than~$n$, then word~$uwv=\tilde{w}$ can be pumped by~$\tilde{w}_i$ (or $\tilde{w}_{i+1}$, \resp).
	Again this contradicts the assumption that~$w$ is a witness for~$\mps(L^*)\ge 2n$.
	Due to the fact that~$|w|\ge2n-1$ one of the previously described cases must occur.
	In conclusion we have that~$w$ cannot be a witness for~$\mps(L^*)\ge 2n$.
	Therefore,~$\mps(L^*)\le 2n-1$.
	
	Now we prove the reachability of the above-mentioned values for~$k$. 
	We distinguish the cases whether~$n>k$,~$n=k$, or~$n<k$:
	\begin{enumerate}
		\item Case~$n> k$: let~$L_{n,k}=\{\,a^i\mid 0\le i\le n-1\,\}\cup\{b^k\}$ which is a finite
		language and thus~$\mps(L_{n,k})=n$.
		Observe, that~$L_{n,k}^*$ is the language of all words that contain only $b$-blocks with lengths that are divisible by~$k$.
		Therefore the word~$w=b^k$ cannot be pumped by a sub-word of length at most~$k-1$ which implies that~$\mps(L_{n,k}^*)\ge k$.
		Assume there is a word~$w\in\{a,b\}^*$ witnessing~$\mps(L_{n,k}^*)>k$ then there are words~$u,v\in\{a,b\}^*$ such that~$uwv$ cannot be pumped by a sub-word
		of the~$k$-prefix~$y$ of~$w$. Due to the structure of~$L_{n,k}^*$ we know that~$uwv$ can be pumped by a sub-word of~$y$, if~$y$ contains an~$a$ or it contains the sub-word~$b^k$. Since~$|y|= k$ one of the conditions must be fulfilled which implies that~$\mps(L_{n,k})=k$.
		
		\item Case~$k=n$: Let~$L=(a^n)^*=L^*$. Then~$\mps(L)=\mps(L^*)=n=k$.

		\item Case~$n< k$: Let~$L_{n,k}=(a^n)^*\cup (b^{k-n+1})^*$. We have~$\mps(L_{n,k})=n$, since $k-n+1\leq (2n-1)-n+1=n$ due to the fact that~$k\in\{1,2,\dots,2n-1\}$. On the other hand we have that~$L_{n,k}^*$ contains all words which only contain~$a$- and~$b$-blocks whose length are divisible by~$n$ and~$k-n+1$, \resp.
		Therefore the word~$w=a^{n-1}b^{k-n}$ cannot be pumped by a sub-word of length~$n-1+k-n=k-1$, which implies that~$\mps(L_{n,k}^*)\ge k$.
		So we assume that there is a word~$w\in\{a,b\}^*$ witnessing~$\mps(L_{n,k}^*)>k$, which implies that there are words~$u,v\in\{a,b\}^*$ such that~$uwv$ cannot be pumped by a sub-word
		of the~$k$-prefix~$y$ of~$w$. Since~$|y|=k\ge n\ge k-n+1$ the word~$y$ must contain the
		sub-word~$a^n$ or~$b^{k-n+1}$, which implies that~$uwv$ can be pumped by that
		sub-word of~$y$. Since this contradicts the assumption that~$w$ is a
		witness for~$\mps(L_{n,k}^*)> k$ we have that~$\mps(L_{n,k}^*)\leq k$. In summary $\mps(L_{n,k}^*)=k$ as desired.
              \end{enumerate}
              This proves the stated claim.
\end{proof}
	
For the intersection operation it was left open in~\cite{DaJe22},
which numbers are reachable for the pumping constant~\mpl. We close
this gap and show that for~\mpc, \mpl, and \mps\ the same set of
numbers is reachable.

\begin{thm}
	For~$K\in\{\mpl,\mps\}$ we have
	$$
	g^{K}_{\cap}(m,n)=
	\begin{cases}
		\{0\}, & \mbox{if $m=0$ or $n= 0$,}\\
		\{1\}, & \mbox{if $m=n=1$,}\\
		\N_0, & \mbox{otherwise,}         
	\end{cases}
	$$
\end{thm}

\begin{proof}
	Obviously we have that~$L\cap \emptyset=\emptyset\cap L=\emptyset$ for each regular language~$L$.
	Assume that~$L,$ $L'$ are regular languages with~$\mpl(L)=\mpl(L')=1$.
	Given a word~$\tilde{w}\in L\cap L'$ such that~$\tilde{w}$ is a witness for~$\mpl(L\cap L')\ge 2$ then~$\tilde{w}$ cannot be pumped by its first letter.
	On the other side we know that~$\tilde{w}$ can be pumped in~$L$ and in~$L'$ by its first
	letter since~$\mpl(L)=\mpl(L')=1$. This implies that each word we obtain from~$\tilde{w}$ by
	pumping its first letter is in~$L$ and~$L'$; therefore in~$L\cap L'$.
	Hence, we can pump~$\tilde{w}$ in~$L\cap L'$ by its first letter which is a contradiction
	to the assumption on~$\tilde{w}$.
	The previously shown reasoning also applies similarly for~$\mps(L)=\mps(L')=1$ because with
	this property each word in~$L$ and~$L'$ can be pumped by any of its letters.
	The value~$k=0$ is unreachable for~$n=m=1$ because each language~$L$ with~$\mpl(L)=1$ or~$\mps(L)=1$ contains the letter~$\lambda$ due to Lemma~\ref{lem:mps-zero-one} and the remark after Lemma~\ref{lem:pumping}.
	Next we construct languages such that all values~$k\ge 0$ can be achieved in the 
	general case for~$m$ and~$n$. 
	Here we distinguish whether~$k$ is equal to zero, one or an odd or an even value which is at least two---notice that the construction for~$k=1$ also applies for~$m=n=1$:
	\begin{itemize}
		\item For~$k=0$ we define~$L_{m,k}=\{a^{m-1}\}$ and~$L_{n,k}=\{b^{n-1}\}$ which
		are finite languages and therefore fulfill~$\mpl(L_{m,k})=\mps(L_{m,k})=m$ and~$\mpl(L_{n,k})=\mps(L_{n,k})=n$.
		Clearly~$L_{m,k}\cap L_{n,k}=\emptyset$ which provides~$\mpl(\emptyset)=\mps(\emptyset)=0=k$.
		
		\item In the case~$k=1$ we define~$L_{m,k}=\{a^{m-1}\}\cup\{b\}^*$ and~$L_{n,k}=\{c^{n-1}\}\cup\{b\}^*$ which fulfill~$\mpl(L_{m,k})=\mps(L_{m,k})=m$ and~$\mpl(L_{n,k})=\mps(L_{n,k})=n$ because~$a^{m-1}\in L_{m,k}$ and~$c^{n-1}\in L_{n,k}$ are not pumpable by any of their sub-words.
		Obviously we have~$L_{m,k}\cap L_{n,k}=\{b\}^*$ which suffices~$\mpl(\{b\}^*)=\mps(\{b\}^*)=1=k$.

		\item Now we study the case where~$k\ge 2$ is an even integer.
		If~$k\ge 2$ one of the values~$m$ and~$n$ must be at least equal to two.
		Since the intersection of regular languages is symmetric in its arguments we
		assume without loss of generality that~$m\ge 2$ and~$n\ge 1$.
				
		We set~$L_{m,k}=\{c^{m-1}\}\cup\{ba\}^*\{b\}\{ad\}^*\cup\{da\}^*\{d\}$
		and~$L_{n,k}=\{e^{n-1}\}\cup B_{k-2}^{(*)}\{d\}^*$.
		We observe that~$c^{m-1}\in L_{m,k}$ and~$e^{n-1}\in L_{n,k}$ are not pumpable 
		which implies that~$m\le\mpl(L_{m,k})\le \mps(L_{m,k})$ and~$n\le\mpl(L_{n,k})\le \mps(L_{n,k})$.
		Since each word in~$L_{n,k}$ is pumpable by each of its letters except~$e^{n-1}$ we obtain~$n=\mpl(L_{n,k})=\mps(L_{n,k})$.
		Further each word in~$\tilde{w}\in\{ba\}^*\{b\}\{ad\}^*\cup \{da\}^*\{d\}$
		is pumpable by a sub-word~$y$ of each sub-word~$w$ of~$\tilde{w}$ with~$|w|\ge 2$.
		This can be seen by looking at the different cases for the prefixes of length two of~$w$ which is done next. For the sake of simplicity we assume~$|w|=2$.
		We will give for each case the word~$y$ and then verify that~$\tilde{w}$ can be pumped
		by~$y$ by distinguishing between the words~$\tilde{w}$ which can contain~$w$:
		\begin{itemize}
			\item For~$w=ab$ we can choose~$y=ab$ which is observed by understanding
			that~$\{ba\}^*\{b\}\{ad\}^*=\{b\}\{ab\}^+\{ad\}^*\cup \{b\}\{ad\}^*$.
			
			\item For~$w=ba$ we can choose~$y=ba$.
			First let~$\tilde{w}=(ba)^ib(ad)^j\in \{ba\}^*\{b\}\{ad\}^*$.
			If~$\tilde{w}=(ba)^{i'}w(ba)^{i''}b(ad)^j$ for~$i=i'+i''+1$ then we obtain by
			pumping~$\tilde{w}$ \textit{via}~$y$ a word~$(ba)^{i+\ell}b(ad)^*$ for~$-1\le \ell$.
			For~$\tilde{w}=(ba)^iwd(ad)^{j-1}$ we obtain by pumping~$\tilde{w}$ \textit{via}~$y$ a word~$(ba)^{i+\ell}b(ad)^*$ for~$0\le \ell$ and the word~$d(ad)^{j-1}=(da)^{j-1}d\in\{da\}^*\{d\}$ for~$\ell=-1$.
			
			\item For~$w=ad$ we can choose~$y=ad$ which is easy to confirm for
			each word in~$\{ba\}^*\{b\}\{ad\}^*$ and each word in~$\{da\}^*\{d\}=\{d\}\{ad\}^*$.
			
			\item For~$w=da$ we can choose~$y=da$ because~$\{ba\}^*\{b\}\{ad\}^*=\{ba\}^*\{b\}\{a\}\{da\}^*\{d\}\cup\{ba\}^*\{b\}$ and for the words in~$\{ba\}^*\{b\}\{a\}\{da\}^*\{d\}\cup\{da\}^*\{d\}$ it is obvious that they can be pumped by~$y=da$.
		\end{itemize}
		In conclusion each word~$\tilde{w}$ in~$L_{m,k}$ can be pumped by a sub-word~$y$ of
		every sub-word~$w$ of~$\tilde{w}$ if~$|w|\ge 2$.
		Therefore we obtain that~$\mpl(L_{m,k})=\mps(L_{m,k})=m$.

		We observe that~$L_{m,k}\cap L_{n,k}=(\{c^{m-1}\}\cup\{ba\}^*\{b\}\{ad\}^*\cup\{da\}^*\{d\})\cap (\{e^{n-1}\}\cup B_{k-2}^*\{d\}^*)=\{\,(ba)^id\mid 0\le i\le (k-2)/2\,\}$
		which is a finite language and therefore suffices~$\mpl(L_{m,k}\cap L_{n,k})=\mps(L_{m,k}\cap L_{n,k})=k$. 
		This is due to the fact that the longest word in this language is~$\tilde{w}=(ba)^{(k-2)/2}d$ which fulfills~$|\tilde{w}|=2\cdot(k-2)/2+1=k-1$.
	
		\item In the case that~$k\ge 2$ is an odd integer we adapt the language~$L_{m,k}$ shown in the previous case to be equal to~$\{c^{m-1}\}\cup\{ba\}^*\{bd\}^*\cup\{da\}^*\{d\}$.
		Indeed the property~$\mpl(L_{m,k})=\mps(L_{m,k})=m$ can be
		proven in the same style as in the previous case.
		Additionally we have~$L_{m,k}\cap L_{n,k}=(\{c^{m-1}\}\cup\{ba\}^*\{bd\}^*\cup\{da\}^*\{d\})\cap (\{e^{n-1}\}\cup B_{k-2}^*\{d\}^*)=\{\,(ba)^ibd\mid 0\le i\le (k-3)/2\,\}$,
		which is a finite language and therefore suffices~$\mpl(L_{m,k}\cap L_{n,k})=\mps(L_{m,k}\cap L_{n,k})=k$. 
		Here the longest word in this language is~$\tilde{w}=(ba)^{(k-3)/2}bd$ which fulfills~$|\tilde{w}|=2\cdot(k-3)/2+2=k-1$.
		
	\end{itemize}
\end{proof}

\bibliographystyle{eptcs} 
\bibliography{shortmarkus,markus}

\end{document}